\UseRawInputEncoding
\documentclass[pra,aps,reprint,superscriptaddress,showpacs,showkeys,nofootinbib]{revtex4-1}
\usepackage{amsmath,amsthm,amssymb,graphicx,subfigure,xcolor}
\usepackage{graphics,float}
\usepackage{color}
\usepackage[unicode=true]{hyperref}
\hypersetup{
     colorlinks=true,       		
     linkcolor=blue,          	
     citecolor=red,            
     urlcolor=magenta,           	
 }
\newtheorem{theorem}{Theorem}
\newtheorem*{theorem*}{Theorem}
\newtheorem{corollary}{Corollary}
\newtheorem*{corollary*}{Corollary}
\newtheorem{lemma}{Lemma}
\newtheorem{example}{Example}
\newtheorem*{lemma*}{Lemma}

\newtheorem*{proposition*}{Proposition}
\theoremstyle{definition}
\newtheorem{definition}{Definition}
\newtheorem*{definition*}{Definition}
\theoremstyle{remark}

\newtheorem*{remark*}{Remark}

\begin{document}

\title{Quantum information masking of an arbitrary qudit can be realized in  multipartite lower dimensional systems}
\author{Wei-Min Shang}
\affiliation{School of Science, Tianjin Chengjian University, Tianjin 300384, People's Republic of China}
\affiliation{Theoretical Physics Division, Chern Institute of Mathematics, Nankai University, Tianjin 300071, People's Republic of China}

\author{Xing-Yan Fan}
\affiliation{Theoretical Physics Division, Chern Institute of Mathematics, Nankai University, Tianjin 300071, People's Republic of China}

\author{Fu-Lin Zhang}
\email{flzhang@tju.edu.cn}
\affiliation{Department of Physics, School of science, Tianjin University, Tianjin 300072, People's Republic of China}

\author{Jing-Ling Chen}
\email{chenjl@nankai.edu.cn}
\affiliation{Theoretical Physics Division, Chern Institute of Mathematics, Nankai University, Tianjin 300071, People's Republic of China}

\date{\today}

\begin{abstract}
Quantum information masking is a protocol that hides the original quantum information from subsystems and spreads it over quantum correlation, which is available to multipartite except bipartite systems. In this work, we explicitly study the quantum information masking in multipartite scenario and prove that all the $k$-level quantum states can be masked into a $m$-qudit systems ($m\geq4$) whose local dimension $d\leq k$ and the upper bound of $k$ is tighter than the quantum Singleton bound. In order to observe the masking process intuitively, explicitly controlled operations are provided. Our scheme well demonstrates the abundance of quantum correlation between multipartite quantum system and has potential application in the security of quantum information processing.

\end{abstract}


\maketitle
\section{Introduction}
In quantum information theory, since the information processing process adheres to the unitary evolution and linear superposition principle, several operations that can be complemented in the classical information process are prohibited in a closed physical system. The notion that reveals these phenomenons is called the ``no-go" theorem. For instance, there exists no universal cloning machine that can replicate an arbitrary unknown pure quantum state, which is known as the no-cloning theorem\cite{WW,D,HY}. A contrary version of the no-cloning theorem states that it is impossible to delete one of two copied arbitrary unknown quantum states without affecting the other in a closed physical system which resulting in no-deleting theorem\cite{AS}. With the in-depth study of quantum information theory, more and more no-go theorems have been mooted like no-broadcasting theorem\cite{HCC,AI}, no-superposition theory\cite{MAM,S,MFV}, no-hiding theory\cite{SA}. These theorems shed light on the discrepancies between quantum mechanics and classical physics from the perspective of information theory, and also root in the security of quantum information processing tasks such as quantum secret sharing\cite{MVA,RD,ZYZ}, quantum key distribution\cite{C,NG} and quantum teleportation\cite{CG,LZ,DP}, et al.

In 2018,  a new no-go theorem named no-masking theorem was proposed by  Kavan Modi et al. which states that it's impossible to hide the original arbitrary unknown quantum states into the quantum correlation between bipartite quantum systems and make it inaccessible to marginal systems\cite{KA}. Moreover, this result not only provides a broader vision of the quantum bit commitment named quantum qubit commitment\cite{D,HH} but also demonstrates its potential application in quantum secret sharing.

This topic attracted wide attention and a series of significant discussions were published.  When the masked set of states is all the states in a quantum system, there is evidence that the no-masking theorem is robust, that is,  even if replacing unitary transformations with general linear transformations, it's impossible probabilistically to mask arbitrary unknown quantum state into a bipartite quantum system\cite{MK}. Interesting, no-masking theorem is available even the set of states is reduced to the set of nonzero measure\cite{XBS}. This finding further enriches the no-masking theorem. Similar to the no-cloning theorem, when the masked set consists of mutual orthogonal (or linearly independent) states, it can be deterministically (or probabilistically) masked into bipartite systems\cite{BS}. Different from the no-cloning theorem, the maskable set of states is abundant. In geometric, the maskable set of states in high dimensional systems is described by the one on two or more hyperdisks\cite{DH}. This result could simplified to qubit case, where the maskable states located on a spherical circle on the Bloch sphere which has been verified experimentally\cite{XB,ZXK}. In addition, there exists a masker for the quantum information encoding into the states described by real destiny matrices\cite{RZZ,ZHJ}.

The above conclusions well describe the situation of quantum information masking in a bipartite system. How does it behave in a multipartite systems? Li and Wang proved that it's possible to mask all the states of a quantum system into a multipartite systems and observe no information in each local system\cite{MY}. Meanwhile, several schemes were proposed, namely, $\mathbb{C}^{d}\rightarrow (\mathbb{C}^{d})^{\otimes 2d}$, $\mathbb{C}^{d}\rightarrow \mathbb{C}^{d}\otimes\mathbb{C}^{d}\otimes\mathbb{C}^{d}$ when $d$ is odd and $\mathbb{C}^{d}\rightarrow \mathbb{C}^{d+1}\otimes\mathbb{C}^{d+1}\otimes\mathbb{C}^{d+1}$ when $d$ is even respectively. An operation like this that can mask all quantum states of a system is called an universal masker. This finding opens up a new direction for quantum information masking. Recently, it's upgraded version named $k$-uniform quantum information masking was proposed which reveals the strong relationship between quantum information masking and quantum error-correcting codes (QECCs)in heterogeneous systems\cite{FML}. A similar topics also were discussed \cite{KZH,WFC}. In addition, the close relationship between quantum multipartite masker and quantum teleportation was studied and a new masking scheme ($\mathbb{C}^{d}\rightarrow \mathbb{C}^{d}\otimes\mathbb{C}^{d}\otimes\mathbb{C}^{d}\otimes\mathbb{C}^{d}$) was provided\cite{WF}. These conclusions all convey a signal that quantum information masking has wide prospect in the security of quantum information process.

So far, the states in the maskable set have been masked into multipartite systems whose local dimension is not lower than that of the system in which they are located. Generally, the operations allowed are simpler in lower dimensional quantum systems especially in experiments. In this paper, we show that it's possible to mask arbitrary  unknown quantum states into a multipartite systems with lower local systems. To facilitate the understanding of the masking process,  vivid controlled operations in four-partite systems was also provided. In addition, we prove that the bound of the dimension of the maskable quantum states is tighter than the famous quantum Singleton bound. This reveals the difference between quantum information masking and quantum error-correcting codes.

\section{PRELIMINARY}
Before starting our main results, the concepts to be used later are sorted out. Firstly, let's have a briefly review of the quantum information masking in multipartite scenario\cite{MY}.
\begin{definition}
An operation S is a multipartite masker which mask the quantum information encoded in the set of states $\{|\psi_{i}\rangle\in H_{A_{1}}\}$ into the set of multipartite states $\{|\Psi_{i}\rangle\in \otimes_{l=1}^{N}H_{A_{l}}\}$ such that all the reductions to each local system of $|\Psi_{i}\rangle$ are identical; i,e., for any $l\in \{1,2,\cdots,N\}$ ,
\begin{equation}
\begin{aligned}
\rho_{A_{l}}={\rm Tr}_{A^{c}_{l}}|\Psi_{i}\rangle\langle\Psi_{i}|.
\end{aligned}
\end{equation}
contain no information about the value of  $i$. Where $A^{c}_{l}$ denotes the complementary set of the set $\{A_{l}\}$ for the set $B=\{A_{1},A_{2},\cdots,A_{N}\}$, i.e., $A^{c}_{l}=B\backslash \{A_{l}\}$.
\end{definition}

Then, a necessary definition is given as follow
\begin{definition}\label{max}
 The set of pure states $\{|\psi_{i}\rangle\}$ called a  "maximum entangled basis" (MEB) in N-qudit systems ($\mathbb{C}^{d})^{\otimes N}$, if it satisfies the conditions as follow

(i)~~For arbitrary pure state  $|\psi_{i}\rangle$, it's reductions to each party are
\begin{equation}
\begin{aligned}
\rho_{A_{l}}={\rm Tr}_{A^{c}_{l}}|\psi_{i}\rangle\langle\psi_{i}|=\frac{I}{d};
\end{aligned}
\end{equation}

(ii)~~The set of pure states $\{|\psi_{i}\rangle\}$ forms an orthonormal basis for the Hilbert space $(H^{d})^{\otimes N}$.
\end{definition}

In addition, to prove the existence of the universal multipartite masker which described by a unitary operator U , we notice the following fact\cite{LG}
\begin{lemma}\label{lemma}
 If two set of states $\{|\psi_{i}\rangle\}_{i=1}^{n}$ and $|\Psi_{i}\rangle_{i=1}^{n}$ satisfy the condition
\begin{equation}
\begin{aligned}
\langle\psi_{i}|\psi_{j}\rangle=\langle\Psi_{i}|\Psi_{j}\rangle,
\end{aligned}
\end{equation}
for $i,j=1,2,\cdots,n$, there always exists  a unitary operator $U$ to complete the following transformation
\begin{equation}
\begin{aligned}
U|\psi_{i}\rangle=|\Psi_{i}\rangle,
\end{aligned}
\end{equation}
where $i=1,2,\cdots,n$.
\end{lemma}

\section{multi-qubit masker}
As is shown in Re.\cite{MY}, arbitrary unknown qubit states $|\psi\rangle=a_{0}|0\rangle+a_{1}|1\rangle$ can be deterministically masked into 4-qubit systems which reads as
\begin{equation}
\begin{aligned}
&|0\rangle\rightarrow \frac{1}{2}(|00\rangle+|11\rangle)\otimes(|00\rangle+|11\rangle),\\
&|1\rangle\rightarrow \frac{1}{2}(|00\rangle-|11\rangle)\otimes(|00\rangle-|11\rangle).\\
\end{aligned}
\end{equation}

Nevertheless, we find the masking capacity of 4-qubit systems is not limited to this.  Next, we give an example to illustrate that it is possible to mask all the states of 4-level quantum system into a 4-qubit system.

\begin{example}
For arbitrary unknown quantum states $|\phi\rangle$ in the system $\mathbb{C}^{4}$, namely
\begin{equation}\label{qu4it}
\begin{aligned}
|\phi\rangle=a_{0}|0\rangle+a_{1}|1\rangle+a_{2}|2\rangle+a_{3}|3\rangle,
\end{aligned}
\end{equation}
where $|a_{0}|^{2}+|a_{1}|^{2}+|a_{2}|^{2}+|a_{3}|^{2}=1$.
According to the lemma $\ref{lemma}$, there exists a masking processing defined by
\begin{equation}
\begin{aligned}
&|0\rangle\rightarrow \frac{1}{2}(|00\rangle+|11\rangle)\otimes (|00\rangle+|11\rangle),\\
&|1\rangle\rightarrow \frac{1}{2}(|00\rangle-|11\rangle)\otimes (|00\rangle-|11\rangle),\\
&|2\rangle\rightarrow \frac{1}{2}(|01\rangle+|10\rangle)\otimes (|01\rangle+|10\rangle),\\
&|3\rangle\rightarrow \frac{1}{2}(|01\rangle-|10\rangle)\otimes (|01\rangle-|10\rangle).\\
\end{aligned}
\end{equation}
\end{example}
\begin{proof}
The total final states in 4-qubit system is
\begin{equation}
\begin{aligned}
|\Phi\rangle=&\frac{1}{2}(a_{0}+a_{1})(|0000\rangle+|1111\rangle)\\
            &+(a_{0}-a_{1})(|0011\rangle+|1100\rangle)\\
            &+(a_{2}+a_{2})(|0101\rangle+|1010\rangle)\\
            &+(a_{2}-a_{3})(|0110\rangle+|1001\rangle).\\
\end{aligned}
\end{equation}
After calculation, the reduction to each local system is derived as follow
\begin{equation}
\begin{aligned}
\rho_{i}=&\frac{1}{4}(|a_{0}+a_{1}|^{2}+|a_{0}-a_{1}|^{2}\\
&+|a_{2}+a_{3}|^{2}+|a_{2}-a_{3}|^{2})I\\
=&\frac{I}{2},
\end{aligned}
\end{equation}
for $i=1,2,3,4$. That is, we observe no information from  each local systems of 4-qubit systems which reads that there exists a universal masker for $\mathbb{C}^{4}\rightarrow \mathbb{C}^{2}\otimes \mathbb{C}^{2}\otimes \mathbb{C}^{2}\otimes \mathbb{C}^{2}$.
\end{proof}

Next, in order to reveal the masking details of the processing above, the universal masker described by controlled operation is gave as follow.

For arbitrary unknown quantum states $|\phi\rangle$ in the system $\mathbb{C}^{4}$ given as (\ref{qu4it}), according to lemma $\ref{lemma}$, we can complete the unitary transformation reads as
\begin{equation}
\begin{aligned}
&|0\rangle\rightarrow |00\rangle_{12},~|1\rangle\rightarrow |10\rangle_{12},\\
&|2\rangle\rightarrow |01\rangle_{12},~|3\rangle\rightarrow |11\rangle_{12}.\\
\end{aligned}
\end{equation}
Then the total states in 2-qubit systems is obtained as
\begin{equation}
\begin{aligned}
|\Phi_{0}\rangle=(a_{0}|00\rangle+a_{1}|10\rangle+a_{2}|01\rangle+a_{3}|11\rangle)_{12}.
\end{aligned}
\end{equation}
The marginal states of $|\Phi_{1}\rangle$ are
\begin{equation}
\begin{aligned}
\rho_{1}=\left(
  \begin{array}{cc}
    |a_{0}|^{2}+|a_{2}|^{2} & a_{0}a_{1}^{*}+a_{2}a_{3}^{*} \\
    a_{0}^{*}a_{1}+a_{2}^{*}a_{3} &  |a_{1}|^{2}+|a_{3}|^{2} \\
  \end{array}
\right),\\
\rho_{2}=\left(
  \begin{array}{cc}
    |a_{0}|^{2}+|a_{1}|^{2} & a_{0}a_{2}^{*}+a_{1}a_{3}^{*} \\
    a_{0}^{*}a_{2}+a_{1}^{*}a_{3} &  |a_{2}|^{2}+|a_{3}|^{2} \\
  \end{array}
\right).
\end{aligned}
\end{equation}
Now, it's obvious that each marginal system contains the original quantum information. Then we add two auxiliary systems and proceed in four steps.

First, we do the C-Not $C_{13}$ on $|\Phi_{1}\rangle$, then the total state in 4-qubit systems is
\begin{equation}
\begin{aligned}
|\Phi_{1}\rangle=(a_{0}|000\rangle+a_{1}|101\rangle+a_{2}|010\rangle+a_{3}|111\rangle)_{123}|0\rangle_{4}.
\end{aligned}
\end{equation}
Now, the reduced density matrices of the first three local systems are derived as follows
\begin{equation}
\begin{aligned}
\rho_{1}&=\rho_{3}=\left(
  \begin{array}{cc}
    |a_{0}|^{2}+|a_{2}|^{2} & 0 \\
    0 &  |a_{1}|^{2}+|a_{3}|^{2} \\
  \end{array}
\right),\\
\rho_{2}&=\left(
  \begin{array}{cc}
    |a_{0}|^{2}+|a_{1}|^{2} & a_{0}a_{2}^{*}+a_{1}a_{3}^{*} \\
    a_{0}^{*}a_{2}+a_{1}^{*}a_{3} &  |a_{2}|^{2}+|a_{3}|^{2} \\
  \end{array}
\right).\\
\end{aligned}
\end{equation}
It is apparent that the original information contained in the off-diagonal elements of the reduced density matrices of  1 and 3 systems is masked.

Second, we do the C-Not operation $C_{24}$ on the total state $|\Phi_{2}\rangle$ and obtain the state as follow
\begin{equation}
\begin{aligned}
|\Phi_{2}\rangle_{1234}=a_{0}|0000\rangle+a_{1}|1010\rangle+a_{2}|0101\rangle+a_{3}|1111\rangle.
\end{aligned}
\end{equation}
After calculation, we can have the reduction to each local system as follow
\begin{equation}
\begin{aligned}
\rho_{1}=\rho_{3}=\left(
  \begin{array}{cc}
    |a_{0}|^{2}+|a_{2}|^{2} & 0 \\
    0 &  |a_{1}|^{2}+|a_{3}|^{2} \\
  \end{array}
\right),\\
\rho_{2}=\rho_{4}=\left(
  \begin{array}{cc}
    |a_{0}|^{2}+|a_{1}|^{2} & 0\\
    0 &  |a_{2}|^{2}+|a_{3}|^{2} \\
  \end{array}
\right).\\
\end{aligned}
\end{equation}
It can be seen that all the original information contained in the off-diagonal elements of the reduced density matrices of four local systems are vanished.

Third, we send 1-particle through  Hadamard gate, then do the C-Not $C_{12}$ and new total state becomes
\begin{equation}
\begin{aligned}
|\Phi_{3}\rangle_{1234}=&\frac{1}{\sqrt{2}}[a_{0}(|0000\rangle+|1100\rangle)+a_{1}(|0010\rangle-|1110\rangle)\\
                        &+a_{2}(|0101\rangle+|1001\rangle)+a_{3}(|0111\rangle-|1011\rangle)].
\end{aligned}
\end{equation}
Now, four reduced density matrices can be obtained as
\begin{equation}
\begin{aligned}
&\rho_{1}=\rho_{2}=\frac{1}{2}I,\\
&\rho_{3}=\left(
  \begin{array}{cc}
    |a_{0}|^{2}+|a_{2}|^{2} & 0 \\
    0 &  |a_{1}|^{2}+|a_{3}|^{2} \\
  \end{array}
\right),\\
&\rho_{4}=\left(
  \begin{array}{cc}
    |a_{0}|^{2}+|a_{1}|^{2} & 0\\
    0 &  |a_{2}|^{2}+|a_{3}|^{2} \\
  \end{array}
\right).\\
\end{aligned}
\end{equation}
It is distinct that we have no information about the original states in both 1, 2  systems and there is still some information leaked on the diagonal elements of 3,4 system.

Finally,  we send 3-particle  through  Hadamard gate, then perform the C-Not operation $C_{34}$ on $|\Phi_{4}\rangle$ and obtain the final total state
\begin{equation}\label{516}
\begin{aligned}
|\Phi_{4}\rangle_{1234}=&\frac{1}{2}[a_{0}(|0000\rangle+|1111\rangle+|0011\rangle+|1100\rangle)\\
                       &+a_{1}(|0000\rangle+|1111\rangle-|0011\rangle-|1100\rangle)\\
                       &+a_{2}(|0101\rangle+|1010\rangle+|0110\rangle+|1001\rangle)\\
                       &+a_{3}(|0101\rangle+|1010\rangle-|0110\rangle-|1001\rangle)].
\end{aligned}
\end{equation}

It is easy to verify that  $\rho_{i}=\frac{I}{2}, i=1,2,3,4$,  which means we complete the masking process $\mathbb{C}^{4}\rightarrow \mathbb{C}^{2}\otimes \mathbb{C}^{2}\otimes \mathbb{C}^{2}\otimes \mathbb{C}^{2}$.  All the C-Not operations above have the form
\begin{equation}
\begin{aligned}
C_{s,t}=|0\rangle\langle0|\otimes I+|1\rangle\langle1|\otimes \sigma_{x},
\end{aligned}
\end{equation}
where $s=1,3,t=2,4.$
It should note that (\ref{516}) is equivalent to the state given as
\begin{equation}
\begin{aligned}
|\Phi\rangle_{1234}=\frac{1}{2}[&a_{0}(|00\rangle+|11\rangle)\otimes (|00\rangle+|11\rangle)\\
                       &+a_{1}(|00\rangle-|11\rangle)\otimes (|00\rangle-|11\rangle)\\
                       &+a_{2}(|01\rangle+|10\rangle)\otimes (|01\rangle+|10\rangle)\\
                       &+a_{1}(|01\rangle-|10\rangle)\otimes (|01\rangle-|10\rangle)].
\end{aligned}
\end{equation}

In order to further explore the masking capacity of multi-qubit systems we give another example
\begin{example}
 For arbitrary unknown quantum states $|\phi\rangle$ in the system $\mathbb{C}^{8}$
\begin{equation}
\begin{aligned}
|\phi\rangle=\sum_{k=0}^{7}a_{k}|k\rangle.
\end{aligned}
\end{equation}
where $\sum_{k=0}^{7}|a_{k}|^{2}=1$.

According to the lemma $\ref{lemma}$, we can complete the deterministic masking process as follow
\begin{equation}
\begin{aligned}
&|0\rangle\rightarrow \frac{1}{2}(|000\rangle+|111\rangle)\otimes (|000\rangle+|111\rangle),\\
&|1\rangle\rightarrow \frac{1}{2}(|001\rangle+|110\rangle)\otimes (|001\rangle+|110\rangle),\\
&|2\rangle\rightarrow \frac{1}{2}(|010\rangle+|101\rangle)\otimes (|010\rangle+|101\rangle),\\
&|3\rangle\rightarrow \frac{1}{2}(|100\rangle+|011\rangle)\otimes (|100\rangle+|011\rangle),\\
&|4\rangle\rightarrow \frac{1}{2}(|000\rangle-|111\rangle)\otimes (|000\rangle-|111\rangle),\\
&|5\rangle\rightarrow \frac{1}{2}(|001\rangle-|110\rangle)\otimes (|001\rangle-|110\rangle),\\
&|6\rangle\rightarrow \frac{1}{2}(|010\rangle-|101\rangle)\otimes (|010\rangle-|101\rangle),\\
&|7\rangle\rightarrow \frac{1}{2}(|100\rangle-|011\rangle)\otimes (|100\rangle-|011\rangle).\\
\end{aligned}
\end{equation}
\end{example}

\begin{proof}
After the unitary transformation, the total state becomes
\begin{equation}
\begin{aligned}
|\Phi\rangle=&\frac{1}{2}[(a_{0}+a_{4})(|000000\rangle+|111111\rangle)\\
            &+(a_{0}-a_{4})(|000111\rangle+|111000\rangle)\\
            &+(a_{1}+a_{5})(|001001\rangle+|110110\rangle)\\
            &+(a_{1}-a_{5})(|110001\rangle+|001110\rangle)\\
            &+(a_{2}+a_{6})(|010010\rangle+|101101\rangle)\\
            &+(a_{2}-a_{6})(|010101\rangle+|101010\rangle)\\
            &+(a_{3}-a_{7})(|100100\rangle+|011011\rangle)\\
            &+(a_{3}-a_{7})(|100011\rangle+|011100\rangle)].\\
\end{aligned}
\end{equation}
We can verify  the reduced  density matrices of each local system by calculation as
\begin{equation}
\begin{aligned}
\rho_{i}=&\frac{1}{4}(|a_{0}+a_{4}|^{2}+|a_{0}-a_{4}|^{2}+|a_{1}+a_{5}|^{2}\\
        &+|a_{1}-a_{5}|^{2}+|a_{2}+a_{6}|^{2}+|a_{2}-a_{6}|^{2}\\
        &+|a_{3}+a_{7}|^{2}+|a_{3}-a_{7}|^{2})I\\
        =&\frac{I}{2}.
\end{aligned}
\end{equation}
where $i=1,2,3,4,5,6$.

Obviously, we can't observe no information about the original states $|\phi\rangle$ which means we complete the masking process $\mathbb{C}^{8}\rightarrow (\mathbb{C}^{2})^{\otimes 6}$.
\end{proof}

From the facts of two examples exhibited above, it is clear that the maximum dimension of the system in which the masked states located is unique up to the number of  MEB of  n-qubit systems. Thus we can derive the following corollary .
\begin{corollary}
For arbitrary unknown quantum states $|\phi\rangle$ in the system $\mathbb{C}^{d}$
\begin{equation}
\begin{aligned}
|\phi\rangle=\sum_{k=0}^{d-1}a_{k}|k\rangle,
\end{aligned}
\end{equation}
If the dimension of the system $\mathbb{C}^{d}$ satisfys
\begin{equation}
\begin{aligned}
d\leq 2^{n},
\end{aligned}
\end{equation}
there always exists a universal masker, which can mask all the states of $\mathbb{C}^{d}$ into 2n-qubit systems deterministically, namely $\mathbb{C}^{d}\rightarrow (\mathbb{C}^{2})^{\otimes 2n}$.
\end{corollary}

\section{multi-qudit masker}

With the increase of the dimensions of the local systems involved in the masking process, the range of maskable states expands. To move forward, we propose the first general conclusion.

\begin{theorem}
For any positive integer $2\leq h\leq d^{2}$, it's possible masking all the states of $h$-level system $\mathbb{C}^{h}$ into 4-qudit systems. i.e.,$\mathbb{C}^{h}\rightarrow\mathbb{C}^{d}\otimes \mathbb{C}^{d}\otimes\mathbb{C}^{d}\otimes\mathbb{C}^{d}$.
\end{theorem}
\begin{proof}
For an arbitrary unknown quantum state of the $d^{2}$-level system $\mathbb{C}^{d^{2}}$ namely
\begin{equation}
\begin{aligned}
|\psi\rangle=\sum_{k=0}^{d^{2}-1}a_{k}|k\rangle,
\end{aligned}
\end{equation}
where $\sum_{k=0}^{d^{2}-1}|a_{k}|=1$. According to lemma $\ref{lemma}$, there exists a unitary can complete the following transformation
\begin{equation}
\begin{aligned}
|k\rangle\rightarrow |\psi_{k}\rangle\otimes|\psi_{k}\rangle.
\end{aligned}
\end{equation}
 Thus the total states in systems $\mathbb{C}^{d}\otimes \mathbb{C}^{d}\otimes\mathbb{C}^{d}\otimes\mathbb{C}^{d}$ can be written as
\begin{equation}\label{G}
\begin{aligned}
|\Psi\rangle=\sum_{k=0}^{d^{2}-1}a_{k}|\psi_{k}\rangle|\psi_{k}\rangle,
\end{aligned}
\end{equation}
where
\begin{equation}\label{h}
\begin{aligned}
|\psi_{k}\rangle =\frac{1}{d}\sum_{j=0}^{d-1}\omega^{j(k\mod d)}|j\rangle|(j+[k/d])\mod d\rangle .
\end{aligned}
\end{equation}
for $\omega=e^{\frac{2\pi i}{d}}$. It's distinct that $\{|\psi_{k}\rangle\}$ is a MEB of the 2-qudit systems. The reductions to all local systems are derived as
\begin{equation}
\begin{aligned}
\rho_{A_{l}}=&{\rm Tr}_{A_{l}^{c}}(|\Psi\rangle\langle\Psi|)\\
        =&\sum_{k=0}^{d^{2}-1}\sum_{t=0}^{d^{2}-1}|a_{k,t}|^{2}{\rm Tr}_{A_{l}^{c}}|\psi_{k}\rangle\langle\psi_{t}|(\langle\psi_{k}|\psi_{t}\rangle)\delta_{kt}\\
        =&\sum_{k=0}^{d^{2}-1}|a_{k}|^{2}{\rm Tr}_{A_{l}^{c}}|\psi_{k}\rangle\langle\psi_{k}|(\langle\psi_{k}|\psi_{k}\rangle)\\
        =&\sum_{k=0}^{d^{2}-1}|a_{k}|^{2}\frac{I}{d}\\
        =&\frac{I}{d}.
\end{aligned}
\end{equation}
\end{proof}

Then, we provide the  corresponding controlled operation. For arbitrary unknown states $|\psi_{0}\rangle \in\mathbb{C}^{d^{2}}$
\begin{equation}
\begin{aligned}
|\psi_{0}\rangle=\sum_{k=0}^{d^{2}-1}|k\rangle,
\end{aligned}
\end{equation}
First, the unitary evolution goes as
\begin{equation}
\begin{aligned}
|k\rangle\rightarrow |k \mod d\rangle|[k/d]\rangle,
\end{aligned}
\end{equation}
where  $[k/d]$ denotes the quotient of $k$ divided by  $d$. Next, by adding two auxiliary systems, the total state holds the form as follow
\begin{equation}
\begin{aligned}
|\psi_{1}\rangle=\sum_{k=0}^{d^{2}-1}a_{k}|k \mod d\rangle|[k/d]\rangle|00\rangle.
\end{aligned}
\end{equation}
Second, we obtain the new total states by  performing the controlled operation $C_{13}$
\begin{equation}
\begin{aligned}
|\psi_{2}\rangle=\sum_{k=0}^{d^{2}-1}a_{k}|k \mod d\rangle|[k/d]\rangle|k \mod d\rangle|0\rangle,
\end{aligned}
\end{equation}
where
\begin{equation}
\begin{aligned}
C_{13}=\sum_{k=0}^{d^{2}-1}|k \mod d\rangle\langle k \mod d|\otimes U^{(k \mod d)},
\end{aligned}
\end{equation}
with
\begin{equation}
\begin{aligned}
U=\sum_{j=0}^{d-1}|(j+1)\mod d\rangle\langle j|.
\end{aligned}
\end{equation}
Third, we perform the controlled operation $C_{24}$ on  $|\psi_{2}\rangle$ and obtain the states as follow
\begin{equation}
\begin{aligned}
|\psi_{3}\rangle=\sum_{k=0}^{d^{2}-1}a_{k}|k \mod d\rangle|[k/d]\rangle|k \mod d\rangle|[k/d]\rangle,
\end{aligned}
\end{equation}
where
\begin{equation}
\begin{aligned}
C_{24}=\sum_{k=0}^{d^{2}-1}|[k/d]\rangle\langle [k/d]|\otimes U^{[k/d]}.
\end{aligned}
\end{equation}

We send 1 and 3 particles  through the  general Hadamard gate, namely
\begin{equation}
\begin{aligned}
|j\rangle\rightarrow \frac{1}{\sqrt{d}}\sum_{j=0}^{d-1}\omega^{j}|j\rangle.
\end{aligned}
\end{equation}

Finally, two controlled not gates given as follow are performed on  $|\psi_{3}\rangle$ subsequently
\begin{equation}
\begin{aligned}
&C_{12}=\sum_{j=0}^{d-1}|j\rangle\langle j|\otimes U^{j}\\
&C_{34}=\sum_{i=0}^{d-1}|i\rangle\langle i|\otimes U^{i}.\\
\end{aligned}
\end{equation}

The final total states become
\begin{equation}
\begin{aligned}
|\psi_{4}\rangle=&\frac{1}{d}\sum_{k=0}^{d^{2}-1}\sum_{i=o}^{d-1}\sum_{j=o}^{d-1}a_{k}\omega^{(i+j)(k\mod d)}|j\rangle|(j\\
&+[k/d])\mod d\rangle|i\rangle|(i+[k/d])\mod d\rangle.
\end{aligned}
\end{equation}
It should note that $|\psi_{4}\rangle$ is equivalent to  (\ref{G}).

It can be deduced from the quantum masking in multi-qubit systems that as the number of local systems participating in masking tasks increases, all quantum states in the higher-dimensional system will be deterministically masked. Thus a more
general conclusion can be obtained.

\begin{theorem}
For all the states of the $w$-level system $\mathbb{C}^{w}$ ,
\begin{equation}
\begin{aligned}
|\phi\rangle=\sum_{k=0}^{w-1}a_{k}|k\rangle,
\end{aligned}
\end{equation}
If we mask all the states into m-qudit $(m\geq4)$ systems, the supremum of the level of the system in which all the masked states locate can be derived as
\begin{equation}
\begin{aligned}
w\leq d^{\lfloor\frac{m}{2}\rfloor}.
\end{aligned}
\end{equation}
\end{theorem}

\begin{proof}
Without loss of generality, we divide the m-qudit system into two parts. Let $\{|\psi_{w}\rangle\},\{|\upsilon_{w}\rangle\}$ be the subsets of the MEBs of $(\mathbb{C}^{d})^{\otimes \lfloor\frac{m}{2}\rfloor}$ and $(\mathbb{C}^{d})^{\otimes \lceil\frac{m}{2}\rceil}$ respectively. According to lemma  $\ref{lemma}$, the masking process can be defined by
\begin{equation}
\begin{aligned}
|k\rangle\rightarrow |\psi_{k}\rangle|\upsilon_{k}\rangle.
\end{aligned}
\end{equation}
where $k=0,1,2,\cdots w$. That is,
\begin{equation}
\begin{aligned}
U:|\phi\rangle=\sum_{k=0}^{w-1}a_{k}|k\rangle\rightarrow |\Phi\rangle=\sum_{k=0}^{w-1}a_{k}|\psi_{k}\rangle|\upsilon_{k}\rangle
\end{aligned}
\end{equation}
For the final state $|\Phi\rangle$, the reductions to the local systems of the first $\lfloor\frac{m}{2}\rfloor$ systems are
\begin{equation}
\begin{aligned}
\rho_{A_{l}}=&{\rm Tr}_{\widehat{A_{l}}}(|\Phi\rangle\langle\Phi|)\\
        =&\sum_{k=0}^{w-1}\sum_{t=0}^{w-1}|a_{k,t}|^{2}{\rm Tr}_{\widehat{A_{l}}}|\psi_{k}\rangle\langle\psi_{t}|(\langle|\upsilon_{k}|\upsilon_{t}\rangle)\delta_{kt}\\
        =&\sum_{k=0}^{w-1}|a_{k}|^{2}\rm{Tr}_{\widehat{A_{l}}}|\psi_{k}\rangle\langle\psi_{k}|(\langle|\upsilon_{k}|\upsilon_{k}\rangle)\\
        =&\sum_{k=0}^{w-1}|a_{k}|^{2}\frac{I}{d}\\
        =&\frac{I}{d}.
\end{aligned}
\end{equation}
Where $l\in\{1,2,\cdots,\lfloor\frac{m}{2}\rfloor\}$ and $\widehat{A_{l}}$ denotes the set $ \{A_{1}, A_{2},\cdots,A_{\lfloor\frac{m}{2}\rfloor}\}\setminus\{A_{l}\}$ and $\lfloor . \rfloor (\lceil.\rceil) $ denotes the integer part.

The same result can be obtained for the last $\lceil\frac{m}{2}\rceil$-qudit systems namely $\rho_{A_{\tilde{l}}}=\frac{I}{d}$ for $\tilde{l}\in\{\lceil\frac{m}{2}\rceil, \lceil\frac{m}{2}\rceil+1,\cdots, m\}$. Thus we observe no information about  $|\phi\rangle$ from all the local systems of the $m$-qudit system.
\end{proof}

Specially, an $((m, w, 2))_{d}$ QECC has the quantum Singleton bound\cite{ZS}
\begin{equation}
\begin{aligned}
w\leq d^{m-2},
\end{aligned}
\end{equation}
Where $m\geq4$. It's easy to verify that $d^{\lfloor\frac{m}{2}\rfloor}\leq d^{m-2}$, that is, the maskable bound of $k$ is tighter than the quantum Singleton bound.

This conclusion implies that, in the quantum information masking process, when the dimension of each local system $d$  and  the level $w$ of the system that all the masked states located in  are determined and $w>d$, it can be seen that at least $2t$ qudit systems need to be employed for $ t =\lceil\log_{d}w\rceil$.

\section{Conclusion}
In this paper, we study the structure of the quantum  multipartite masker. Noticing all the known masking schemes, each local system's dimension is not less than the level of the system which all the states to be masked located in. Here we present some schemes to show it's possible masking $k$-level states into a $m$-partite system with local dimension $d$ for $k\geq d,m\geq4$ and obtain a tighter bound of $k$ than the famous quantum Singleton bound of QECC. We also provide the controlled operation in four-partite systems to unveil the mystery of quantum information masking.

Our conclusion provides hands for the question ``can all quantum states of level $d$ be hidden into
tripartite quantum system $\mathbb{C}^{d}\otimes \mathbb{C}^{d}\otimes\mathbb{C}^{d}$ with $n < d$ or not?"  raised in Re.\cite{MY} . By the way, a new question is raised: can our conclusion be generalized to $k$-uniform scenario?
\begin{acknowledgments}
This work was support by the Scientific Research and Innovation Project of Tianjin (No. 2021YJSB072) ,  National Natural Science Foundation of China (Grants Nos. 11675119 ,11575125,12275136, 11875167 and 12075001 ). and the Fundamental Research Funds for the Central Universities (Grant No. 3072022TS2503).
\end{acknowledgments}


\end{document}